\documentclass[5p,sort&compress]{elsarticle}

\usepackage[english]{babel}
\usepackage[colorlinks,linkcolor=blue]{hyperref}
\usepackage{amsfonts}
\usepackage{amsmath}
\usepackage{amssymb}
\usepackage{color}
\usepackage{graphicx}
\usepackage{float}
\usepackage[mathscr]{eucal}
\usepackage{amsfonts}
\usepackage{float}
\usepackage[colorlinks,linkcolor=blue]{hyperref}
\usepackage[latin1]{inputenc}
\usepackage{amsmath}
\usepackage{graphicx}
\usepackage{amssymb}
\usepackage{amsthm}
\usepackage{verbatim}
\usepackage{epsfig}
\usepackage[labelfont=bf]{caption}
\usepackage{enumerate}
\usepackage{subfig}
\usepackage{epstopdf}

\newtheorem{lemma}{Lemma}
\newtheorem{proposition}{Proposition}
\newtheorem{corollary}{Corollary}
\newtheorem{fact}{Fact}

\newtheorem{remark}{Remark}

\newtheorem{assumption}{Assumption}

\def\begcen{\begin{center}}
\def\endcen{\end{center}}

\newcommand{\col}{ \mbox{col} }
\newcommand{\rank}{ \mbox{rank } }

\def\caly{{\cal Y}}

\def\hal{{1 \over 2}}

\def\liminf{\lim_{t \to \infty}}

\def\L2{{\cal L}_2}
\def\L2e{{\cal L}_{2e}}

\def\rea{\mathbb{R}}

\def\sign{\mbox{sign}}
\def\adj{\mbox{adj}}

\def\begmat#1{\begin{bmatrix}#1\end{bmatrix}}
\def\begali#1{\begin{align}{#1}\end{align}}
\def\begalis#1{\begin{align*}{#1}\end{align*}}

\def\begequarr{\begin{eqnarray}}
\def\endequarr{\end{eqnarray}}
\def\begequarrs{\begin{eqnarray*}}
\def\endequarrs{\end{eqnarray*}}
\def\begarr{\begin{array}}
\def\endarr{\end{array}}
\def\begequ{\begin{equation}}
\def\endequ{\end{equation}}
\def\lab{\label}
\def\begdes{\begin{description}}
\def\enddes{\end{description}}
\def\begenu{\begin{enumerate}}
\def\begite{\begin{itemize}}
\def\endite{\end{itemize}}
\def\endenu{\end{enumerate}}

\def\lef[{\left[\begin{array}}
\def\rig]{\end{array}\right]}

\def\begcen{\begin{center}}
\def\endcen{\end{center}}
\def\begrem{\begin{remark}\rm}
\def\endrem{\end{remark}}
\def\begassum{\begin{assumption}}
\def\endassum{\end{assumption}}
\def\begassums{\begin{assumption*}}
\def\endassums{\end{assumption*}}
\def\begassu{\begin{ass}}
\def\endassu{\end{ass}}
\def\beglem{\begin{lemma}}
\def\endlem{\end{lemma}}
\def\begcor{\begin{corollary}}
\def\endcor{\end{corollary}}
\def\begfac{\begin{fact}}
\def\endfac{\end{fact}}
\def\TAC{{\it IEEE Trans. Automat. Contr.}}
\def\AUT{{\it Automatica}}

\def\SCL{{\it Systems and Control Letters}}

\def\liminf{\lim_{t \to \infty}}

\def\L2e{{\cal L}_{2e}}

\def\rea{\mathbb{R}}
\def\intnum{\mathbb{Z}}
\def\sign{\mbox{sign}}

\def\adj{\mbox{adj}}
\def\col{\mbox{col}}
\def\hal{{1 \over 2}}
\def\et{\varepsilon_t}

\def\rank{\mbox{rank}\;}

\def\AJC{{\it Asian Journal of Control}}
\def\ARC{{\it Annual Reviews in Control}}

\def\IJACSP{{\it Int. J. on Adaptive Control and Signal Processing}}

\def\TAC{{\it IEEE Trans. Automatic Control}}

\def\EJC{{\it European Journal of Control}}

\def\SCL{{\it Systems \& Control Letters}}
\def\AUT{{\it Automatica}}

\def\CSM{{\it IEEE Control Systems Magazine}}

\def\CSL{{\it IEEE Control Systems Letters}}

\def\beal#1{\begin{align}{#1}\end{align}}
\def\begali#1{\begin{align}{#1}\end{align}}
\def\begalis#1{\begin{align*}{#1}\end{align*}}

\def\begmat#1{\begin{bmatrix}#1\end{bmatrix}}

\def\reapos{\mathbb{R}_{>0}}

\def\intnum{\mathbb{Z}}

\def\begsubequ{\begin{subequations}}
	\def\endsubequ{\end{subequations}}
\usepackage{xcolor}

\graphicspath{{figs/}}

\input epsf

\begin{document}

\begin{frontmatter}
\title{A Globally Convergent Estimator of the Parameters of the Classical Model of a Continuous Stirred  Tank Reactor}
\author[ITMO]{Anton Pyrkin}\ead{a.pyrkin@gmail.com}
\author[ITMO]{Alexey Bobtsov}\ead{bobtsov@mail.ru}
\author[ITMO,ITAM]{Romeo Ortega}\ead{romeo.ortega@itam.mx}
\author[ITAM]{Jose Guadalupe Romero}\ead{jose.romerovelazquez@itam.mx}
\author[UCL]{Denis Dochain}\ead{denis.dochain@uclouvain.be}
\address[ITMO]{Department of Control Systems and Robotics, ITMO University, Kronverkskiy av. 49, Saint Petersburg, 197101, Russia}
\address[ITAM]{Departamento Acad\'{e}mico de Sistemas Digitales, ITAM, Progreso Tizap\'an 1, Ciudad de M\'exico, 04100, M\'{e}xico}
\address[UCL]{ICTEAM, Universit\'e Catholique de Lovain, Avenue Georges Lem\^aitre 4-6, 1348 Louvain-la-Neuve, Belgium}

\begin{abstract}
In this paper we provide the first solution to the challenging problem of designing a {\em globally exponentially convergent} estimator for the parameters of the standard model of a continuous stirred tank reactor. Because of the presence of {\em non-separable} exponential nonlinearities in the system dynamics that appear in Arrhenius law, none of the existing parameter estimators is able to deal with them in an efficient way and, in spite of many attempts, the problem was open for many years. To establish our result we propose a novel procedure to obtain a suitable {\em nonlinearly parameterized}  regression equation and introduce a {\em radically new} estimation algorithm---derived applying the Immersion and Invariance methodology---that is applicable to these regression equations. A further contribution of the paper is that parameter convergence is guaranteed with weak excitation requirements.   
\end{abstract}
%

\end{frontmatter}

\section{Introduction}
\lab{sec1}
%
The non-adiabatic continuous stirred tank reactor 

\noindent (CSTR) is a common chemical and biochemical system in the process industry, and it is described extensively in \cite{BASDOCbook,SEBEDGMELbook}. To comply with the modern stringent monitoring and control requirements it is necessary to dispose of a reliable model, see \cite{DOC} for a tutorial on their control and parameter estimation. Unfortunately, the dynamics of CSTRs is described by differential equations with highly uncertain parameters and, in particular, containing parameter-dependent {\em exponential terms} appearing in Arrhenius law that describes the behavior of the reaction rate. Since these nonlinearities are {\em not separable}, that is, they cannot be expressed as a product of a function of measurable signals and a function depending only on the parameters, none of the existing parameter estimation techniques is applicable to them. See \cite{ORTetalifac23} for a recent review of the existing approaches to solve this kind of problems. 

The main contribution of this paper is to provide the first globally convergent estimator for the parameters of a CSTR assuming {\em known} only the kinetic constant appearing in Arrhenius law. Towards this end, we introduce a novel procedure to obtain a suitable nonlinearly parameterized  regression equation (NLPRE) and---applying the Immersion and Invariance (I\&I) methodology \cite{ASTKARORTbook}--- propose a {\em radically new} estimation algorithm applicable to these regression equations. An additional contribution of the paper is the fact that parameter convergence is guaranteed with extremely weak excitation requirements---namely {\em interval exictation}  \cite{KRERIE}.   

We underscore the fact that to establish our result we do not assume that the parameters leave in known compact sets, that the nonlinearities satisfy some Lipschitzian properties, nor rely on injection of high-gain---via sliding modes of the use of fractional powers---or the use of complex, computationally demanding methodologies. Instead, we propose to design a classical on-line estimator whose dynamics is described by an ordinary differential equation given in a compact precise form.\\

\noindent {\bf Notation.} $\rea_+$ and  $\intnum_+$ denotes the positive real and integer numbers, respectively, and $\rea^n_+$   the set of $n$-dimensional vectors whose elements are all positive.  For a column vector $a=\col(a_1,a_2,\dots,a_n) \in \rea^n$, we denote $|a|^2:=a^\top a$ and define $a_{i,j}:=\col(a_i,a_{i+1},\dots,a_j)$, for $i,j \in \intnum_+$, with $j>i$. For a matrix $A \in \rea^{n \times m}$ we use $\|A\|$ for its Euclidean norm and denote its elements as $A_{ij}$. The action of an LTI filter $\mathcal H(p) \in \rea(p)$, with $p$ the derivative operator, {\em i.e.}, $p(u)=:{du(t)\over dt}$, on a signal $u(t)$ is denoted as $\mathcal H(p)(u)$. To simplify the notation, the arguments of all functions and mappings are written only when they are first defined and are omitted in the sequel.

%
\section{System Model and Main Result}
\lab{sec2}
%
We consider in the paper a non-isothermal CSTR  with one reactant for which mass and energy balance considerations lead to the following differential equations \cite{DOC,MATSIM}:
\begsubequ
\lab{sys}
\begali{
\lab{dotca}
\dot C_A & = {q\over V}(C_{in} - C_A) - k_0 e^{-{E\over RT}}C_A \\
\lab{dott}
\dot T & = {q\over V}(T_{in} - T) - {\Delta H\over \rho C_p}k_0 e^{-{E\over RT}}C_A   +{h A \over \rho C_p V} (T_w- T)
}
\endsubequ
where the system states $C_A(t) \in \rea_+$ and $T(t) \in \rea_+$ are the concentration of the product and temperature, respectively. $T_w(t) \in \rea_+$ is the heat exchanger temperature, that is an {\em input} signal and $T_{in}(t) \in \rea_+$ is the influent temperature. The  positive parameters\footnote{The parameter $\Delta H$ is negative for exothermic reactions, but to simplify the notation we preserve the qualifier positive for all the parameters.} 
$$
\{E, R, q, V, C_{in}, \Delta H, \rho, C_p, h, A\},
$$
whose physical meaning may be found in \cite{DOC,SEBEDGMELbook}, are {\em unknown}. Notice that the only parameter assumed known is $k_0$, which is  the kinetic constant appearing in Arrhenius law.\footnote{It is argued in \cite{SIMCSTR} that, in practice, it is possible to  estimate this coefficient from bench scale experiments.} 

We find convenient to rewrite the model in the more compact form
\begsubequ
\lab{newsys}
\begali{
	\lab{x1}
	\dot C_A & = \theta_1 -\theta_2 C_A -k_0 e^{-{\theta_5 \over T}}C_A \\
	\lab{x2}
	\dot T & = \theta_2 (T_{in}-T) -\theta_3 k_0 e^{-{\theta_5 \over T}}C_A +\theta_4 u
}
\endsubequ
where we defined
$$
\theta :=\col\left( {q\over V}C_{in}, {q\over V},{\Delta H\over \rho C_p},{h A \over \rho C_p V}, {E\over R}\right) \in \rea^5_+,
$$
and the new input
$$
u := T_w- T.
$$

Our task is to generate a {\em globally convergent  estimate} of the parameters $\theta$. To provide a solution to this problem we need the following.

\begin{assumption}\em
\lab{ass1}
The system states $C_A$, $T$ and the influent temperature $T_{in}$  are {\em measurable}.
\end{assumption}
 
 The main result of the paper is contained in the following proposition.
 
\begin{proposition}\em
\lab{pro1}
Consider the CSTR model \eqref{newsys} verifying Assumption \ref{ass1}. There exists an {\em on-line parameter estimator} of the form
\begalis{
	\dot \chi &= F_\chi(\chi,C_A,T,T_{in},u)\\
	\hat \theta &=H_\chi(\chi,C_A,T,T_{in},u)
}
with $\chi(t) \in \rea^{n_\chi}$ such that we ensure {\em global exponential convergence} of the estimated parameters under a weak excitation assumption. That is, for all $(C_A(0),T(0)) \in \rea_+^2,\;\chi(0) \in \rea^{n_\chi}$ and all continuous $u(t) \in \rea$ that generates a bounded state trajectory we ensure
\begequ
\lab{concon}
\liminf |\tilde \theta(t)| = 0,\quad (exp.)
\endequ
where $\tilde \theta:=\hat \theta - \theta$ is the parameter estimation error, with all signals remaining bounded.
\end{proposition}
%
\section{Proof of Proposition \ref{pro1}}
\lab{sec3}
%
The proof proceeds in two steps, which are described in the following two subsections. First, the construction of an {\em overparameterized} linear regression equation (LRE) for the parameters $\theta_{1,4}$ and their  consistent estimation with the least-squares plus dynamic regressor extension (LS+DREM) estimation algorithm reported in \cite{ORTROMARA,PYRetal}. Second, the creation of a NLPRE for the remaining parameter $\theta_5$ and their estimation with a {\em radically new} estimation procedure based on the I\&I methodology.
\subsection{Estimation of $\theta_{1,4}$} 
\lab{subsec31}
%
First, we eliminate the exponential term combining equations \eqref{x1} and \eqref{x2} as
\beal{
\theta_3 \dot C_A - \dot T = \theta_1\theta_3 - \theta_2\theta_3 C_A -\theta_2 (T_{in}-T)+ \theta_4(- u),
}
that we rearrange as
$$
- \dot T = \theta_1\theta_3+\theta_2 (T-T_{in})+\theta_3 (-\dot C_A)+ \theta_4 (-u) + \theta_2\theta_3 (-C_A).
$$

Fix a constant $\lambda>0$ and apply the LTI, first order filter ${\lambda \over p+\lambda}$, with $p:={d\over dt}$, to the previous equation to get the following  overparameterized LRE
\begali{
\lab{lre}
	y & = \eta^\top \varphi +\et
}
where  we defined
$$
y := {\lambda p\over p+\lambda}(-T),
$$
the parameter vector 
\begali{
\lab{eta}
\eta &:=\col(\theta_1\theta_3,\theta_2,\theta_3,\theta_4,\theta_2\theta_3) \in \rea^5_+,
}
and the regressor vector signal
{\scriptsize
\begali{
\lab{varphi}
\varphi &:=\col\left((1-\et),{\lambda \over p+\lambda} (T-T_{in}), {\lambda p\over p+\lambda}(-C_A),{\lambda \over p+\lambda} (-u),{\lambda \over p+\lambda} (-C_A) \right)
}
}
where $\et$ is an exponentially decaying signal, which is neglected in the sequel. 

To estimate the parameters of the LRE \eqref{lre} we impose the {\em necessary} assumption that it is {\em identifiable}  \cite{GOOSINbook}. That is, that there exists a set of time instants---$\{t_i\}, i=1\dots, 5,\;t_i \in \reapos$, such that
$$
\rank\Big\{\begmat{\varphi(t_1)|\varphi(t_2)|&\cdots&|\varphi(t_5)}\Big\}=5.
$$

We recall the following result of \cite{WANORTBOB}.

\begin{lemma}\em
\lab{lem1}
The LRE \eqref{lre} is identifiable {\em if and only if} the regressor vector $\varphi$ is interval exciting (IE)  \cite{KRERIE}. That is,  {there exist} constants $c_c>0$ and $t_c>0$ such that
$$
\int_0^{t_c} \varphi(s) \varphi^\top(s)  ds \ge c_c I_5.
$$
\end{lemma} 

The estimator we propose below consistently estimates the parameters $\eta_{1,4}$ of the LRE \eqref{lre}. It is clear from \eqref{eta} that this yields the estimates of $\theta_{2,4}$, which are equal to $\eta_{2,4}$. On the other hand, to estimate $\theta_1$ we need to compute it via
\begequ
\lab{hatthe1}
 \hat \theta_1={\hat \eta_1 \over \hat \eta_3}.
\endequ
It will be shown below that the monotonicity property of each parameter estimation error of DREM estimators avoids the possibility of a division by zero in \eqref{hatthe1}.
 
We are in position to present the main result of the subsection whose proof relies on the following.

\begin{assumption}\em
\lab{ass2}
The regressor $\varphi$ given in \eqref{varphi} is IE.
\end{assumption}

\begin{proposition}\em
\lab{pro2}
		Consider the LRE \eqref{lre} verifying Assumption \ref{ass2}. Define the LS+DREM interlaced estimator  with time-varying forgetting factor
		\begalis{
			\dot{\hat \mu  }   & =\alpha F   \varphi   (y  -\varphi^\top    \hat\mu   ),\; \hat\mu  (0)=:\mu ^0  \in \rea^5\\
		\dot {F}  & = -\alpha F  \varphi \varphi^\top    F {+ \beta F },\;F(0)={1 \over f_0} I_5\\
\dot{\hat \eta}_{1,4}   & =\gamma_a \Delta  (\caly_{1,4}  -\Delta  \hat\eta_{1,4}),\; \hat\eta_{1,4}(0)=:\eta^0_{1,4} \in \rea^4_+\\
		\dot z&=-\beta z,\;z(0)=1,
}
where we defined
		\begalis{
		\beta&:=\beta_0\Big(1-{\|F\| \over M}\Big)\\
			\Delta  & :=\det\{I_5- z  f_0F\}\\
			\caly  & := \adj\{I_5- z  f_0F\}[\hat\mu - z f_0 F \mu^0 ],
		}	
with tuning gains $\alpha>0$, $f_0>0,\; {\beta_0 > 0},\;M \geq {1 \over f_0}$ and $\gamma_a >0$. Define the parameter estimates $\hat\theta_{2,4}=\hat\eta_{2,4}$ and   $\hat \theta_1$ given by \eqref{hatthe1}.   Then, for all $f_0>0$,  $\eta_1^{0}$ and $\theta_{2,4}^{0}$, we have that 
$$
\liminf |\tilde \theta_{1,4}(t)|=0,\quad (exp.)
$$
with all signals {\it bounded}, where $\tilde \theta_{1,4}:=\hat \theta_{1,4}- \theta_{1,4}$ is the estimation error vector. 
\end{proposition}
\begin{proof}
The proof of parameter convergence follows immediately from Lemma \ref{lem1} and \cite[Corollary 1]{ORTROMARA}. The proof that $\hat \theta_3(t) > 0$ to avoid a singularity in the calculation of \eqref{hatthe1} follows from the fact that the parameter errors verify the {\em monotonicity} condition
$$
|\tilde \theta_i(t_a)| \leq |\tilde \theta_i(t_b)|,\;\forall\, t_a \geq t_b \geq 0,\;i=2,3,4.
$$
\end{proof}
\subsection{Estimation of $\theta_5$ assuming known $\theta_{2,4}$} 
\lab{subsec32}
%
As in the previous subsection we first need to express  the parameter  $\theta_5$ in a suitable equation, which is now nonlinearly parameterized. Towards this end, we find convenient to rewrite   \eqref{x2} in the compact form
\begequ
\lab{newregequ}
\dot T = \zeta_1 - \zeta_2 e^{\theta_5  \phi},
\endequ
with the definitions
\begalis{
\zeta_1 & := \theta_2 (T_{in}-T)  +\theta_4 u\\
\zeta_2 & := \theta_3 k_0C_A\\
\phi &:= -{1 \over T}.
}

We propose in the following lemma a {\em completely new} estimator for the parameter $\theta_5$ of equation \eqref{newregequ}, which is obtained applying the I\&I methodology \cite{ASTKARORTbook}. For the sake of clarity, we present first a lemma for the {\em ideal} case where we assume that the parameters  $\theta_{2,4}$ are known---denoting the signals of this estimator with $(\cdot)^\star$. Then, as a corollary we define the actual $\hat \theta_5$ that uses the signals $\hat \theta_{2,4}$ and, carrying a simple perturbation analysis to the scheme of the lemma, prove its convergence.

\begin{lemma}\em
\lab{lem2}
Consider the equation \eqref{newregequ} assuming known the parameters $\theta_{2,4}$---hence the signals $\zeta_1$ and $\zeta_2$ are {\em measurable}. Following the I\&I methodology define the estimate of $\theta_5$ as the sum of a proportional and an integral term as
\begequ
\lab{hatrho}
{\hat\theta_5}^\star = {\rho }_P(T) +{\rho }_I ^\star,
\endequ
where 
\begsubequ
\lab{pi}
\begali{
	\lab{rho}
	{\rho }_P(T) & = \gamma_b \ln(T) \\
	\lab{rhoi}
	{\dot \rho }^\star_I & = \gamma_b \left( \zeta_1 -  \zeta_2 e^{ {\hat\theta_5}^\star \phi}\right)\phi
}
\endsubequ
with $\gamma_b$ a tuning gain such that $\gamma_b \sign(\theta_3) >0$.\footnote{This requirement is imposed because, as mentioned in Footnote 1, the parameter $\theta_3$ is negative for exothermic reactions and, in the sequel, we will require $\gamma_b \theta_3 >0$.} Then, for all initial conditions $\rho  ^\star_I(0) \in \rea$ we have that  all signals are bounded and
$$
\liminf  |{\tilde\theta_5^\star}(t)|=0,
$$ 
where ${\tilde\theta_5^\star}:={\hat\theta_5^\star} -\theta_5 $ is the parameter error of the ideal estimator. 
\end{lemma}
\begin{proof}
The dynamics of the parameter error   is given by
\begalis{
\dot{\tilde\theta}_5^\star &=\dot{\hat\theta}_5^\star 
\\
&= \dot{\rho  }_P + {\dot \rho }^\star_I
\\
&= \gamma_b{1 \over T} \dot T  +  \gamma_b \left( \zeta_1 -  \zeta_2 e^{\hat \theta _5^\star \phi}\right)\phi\\
&= -\gamma_b  ( \zeta_1 - \zeta_2 e^{\theta_5  \phi})\phi +  \gamma_b \left( \zeta_1 -  \zeta_2 e^{\hat \theta_5^\star \phi}\right)\phi\\
&= \gamma_b \zeta_2 \left( e^{\theta_5  \phi} - e^{\hat \theta _5^\star \phi} \right)\phi\\
&= \gamma_b \zeta_2 \left( e^{(\hat \theta _5^\star - \tilde \theta _5^\star)\phi} - e^{\hat \theta_5^\star \phi} \right)\phi\\
&= \gamma_b \zeta_2 e^{\hat \theta _5^\star \phi} \left( e^{- \tilde \theta_5^\star \phi} - 1 \right)\phi\\
&=: m \left( e^{- \tilde \theta_5^\star \phi} - 1 \right)\phi,
}
where we defined the signal $m(t) \in \rea_+$, that satisfies 
$$
m(t):=\gamma_b \zeta_2(t) e^{\hat \theta_5^\star(t) \phi(t)}>0,\;\forall t \geq 0.
$$ 

To analyze the stability of the error equation consider the Lyapunov function $V(\tilde \theta_5^\star)=\hal (\tilde\theta_5^\star)^2$, whose derivative yields
\begalis{
\dot V &= m \left( e^{- \tilde \theta_5^\star \phi} - 1 \right)\tilde\theta_5^\star \phi< 0,\quad \forall \; \tilde\theta_5^\star \neq 0.
}
The claim of negativity stems from the facts that the function 
$$
f(z):=z(e^{-z}-1)<0,\; \forall z\neq 0,
$$
and that $\phi(t)=-{1 \over T(t)} \neq 0,\;\forall t \geq 0$. This completes the proof.
\end{proof}

\subsection{Certainty equivalent estimation of  $\theta_5$} 
\lab{subsec33}
%
In this subsection we present the actual estimator of the parameter $\theta_5$ that relies on the estimation of $\theta_{2,4}$ as explained in Proposition \ref{pro2} and the {\em ad-hoc} application of the certainty equivalent principle. As will be seen below the proof boils down to a simple perturbation analysis.
\begin{proposition}\em
\lab{pro3}
Consider the NLPRE \eqref{newregequ} and the I\&I estimator of $\theta_5$ as the sum of a proportional and an integral term as
\begequ
\lab{hatrho1}
\hat  \theta_5 = \rho  _P(T) +{\rho }_I ,
\endequ
where $\rho  _P(T)$ is given in \eqref{rho},\footnote{Notice that the proportional term $\rho  _P(T)$ is independent  of the parameters.}
\begali{
\lab{dotrhoi1}
	\dot {\rho }_I  & = \gamma_b \left(\hat \theta_2 (T_{in}-T)  + \hat \theta_4 u-  \hat \theta_3 C_A\, e^{\hat \theta_5  \phi}\right)\phi,
}
with $\gamma_b$ such that $\gamma_b \sign(\theta_3) >0$ and the estimated parameters $\hat \theta_{2,4}$ generated as indicated in Proposition \ref{pro2}. Then, for all initial conditions ${\rho }_I (0) \in \rea$ we have that all signals are bounded and the parameter error verifies
$$
\liminf |\tilde \theta_5(t)|=0. 
$$ 
\end{proposition}	
\begin{proof}
From the proof of Lemma \ref{lem2} it is clear that the claim will be established if we can prove that $\dot {\rho }_I (t) \to \dot \rho  ^\star_I(t)$. Towards this end, we write \eqref{dotrhoi1} in the form
\small{
\begalis{
	\dot {\rho }_I  & = \gamma_b \left[(\theta_2 + \tilde \theta_2) (T-T_{in})  + (\theta_4+\tilde \theta_4) u  - (\theta_3 +\tilde \theta_3) C_A\, e^{\hat \theta_5 \phi}\right]\phi \\
 & =\dot \rho ^\star_I +  \begmat{ T-T_{in} & -C_A\, e^{\hat \theta_5  \phi} &  u} \tilde \theta_{2,4} \phi,
}
}
\normalsize{where the second right hand term is bounded and exponentially converges to zero. This completes the proof.}

\end{proof}
\section{Simulation Results}
\lab{sec4}
%
In this section we present simulations of the proposed estimators  using the parameters of Table \ref{tab}, which are taken from \cite{SIMCSTR}. 

The simulations where carried out under the following three considerations. 

\begenu

\item The input signal $T_w$ is chosen time varying, such that it provides  sufficient excitation to the system to ensure parameter convergence. Solid arguments about the use of sinusoidal input signals  are presented in \cite{BRAetal} for the particular case of process control systems and in \cite{SASBODbook} for general systems. 
 
\item The temperature $T_{in}$ is normally assumed constant, however in the simulations we consider---again, for excitation requirement---that it is subject to small step changes \cite{SIMCSTR}. 

\item Although in practice, one can determinate the pre--exponential non thermal factor $k_0$ from bench scale experiments \cite{SIMCSTR}, we present simulations of the estimator of Proposition \ref{pro2} considering that $k_0$ is wrongly estimated.
\endenu

{\tiny
\begin{table}[htbp]
 \begin{center}
\begin{tabular}{ |c|c|c|c| }
\hline
\tiny Symbol &  \tiny Value & \tiny Symbol & \tiny Value\\
\hline
\hline
$q$ &\tiny 1 $[m^3/hr]$ &$R$ &  \tiny 1.98589 $[kcal/kgmol K]$ \\
\hline
$V$ & \tiny 1 $[m^3]$ &$\Delta H$ &  \tiny -$5960$ $[kcal/kgmol]$ \\
\hline
$k_0$& \tiny 3.5$\cdot 10^{7}$ $[1/hr]$  &$\rho C_p$& \tiny  480 $[kcal/(m^3 K)]$  \\
\hline
$E$ & \tiny 11850 $[kcal/kgmol]$ & $hA$ & \tiny  145 $[kcal/(K hr)]$  \\
\hline
$C_{in}$ & \tiny 10 $[kgmol/m^{3}]$ &  & \tiny    \\
\hline
\end{tabular}
\caption{Parameters of the system.}
\end{center}
\end{table}
 \label{tab}
}

To evaluate the performances of both estimators we carry out simulations  under different values of $\gamma_a$, different values of $\gamma_b$ and the last one assuming that  the  factor $k_0$ is measurable with an error of up to $\pm 20\%$  with respect to its real value. 

In all simulations we fix the  temperature $T_{in}$ as,

\begin{eqnarray}
\begin{aligned}
    T_{in}(t)& =
    \left\{
    \begin{aligned}
         297 & &  t\in[0,7) 
         \\
         299 & &  t \in  [7,12]  \\
           298 & & t >12
               \end{aligned}
    \right. 
    \end{aligned}
\label{Tin}
\end{eqnarray}
notice that the temperature changes are of very small amplitude. The initial conditions were taken as $C_A(0)=0.5$, $T(0)=120$,  the initial values of the parameter estimators were taken as $\mu^0= 0.1$, $\hat \eta_1^0=\hat \theta^0_{2,4}=-1$, $\rho_I(0)=7800$. We chose the filter parameter as $\lambda=1$ and the tuning gains of the estimator as $f_0=0.1$, $\alpha=68.4$, $\beta_0=30.6$ and $M=130.5$.
 
 For the first simulation we  propose 
 \begequ
 \lab{tw}
 T_w(t)=350-20\exp(-0.001t)\cos(4t),
 \endequ
 and different  values of $\gamma_a$. The transient behavior of the  estimation errors for each $\gamma_a$ is identified by the color in Fig. \ref{Tw1}. For all adaptation gains the response is quite smooth and, as predicted by the theory, the rise time diminishes with increasing gains. Also, we notice that the response of $\hat \theta_5$ is significantly slower\footnote{Notice the difference in time scales.} than the other four parameters, which stems from the fact that---as discussed in Subsection \ref{subsec33}---the correct estimation of $\hat \theta_5$ depends on the convergence to zero of $\tilde \theta_{1,4}$. Also, notice the presence of an unexpected small oscillation around $4$hrs, that may be due to numerical inaccuracies. Other signals $T_w$ were tested, {\em e.g.}, without exponential decay and different frequencies, observing a similar behavior in all cases.

\begin{figure}[!htp]
 \centering
  {
    \includegraphics[width=0.45\textwidth]{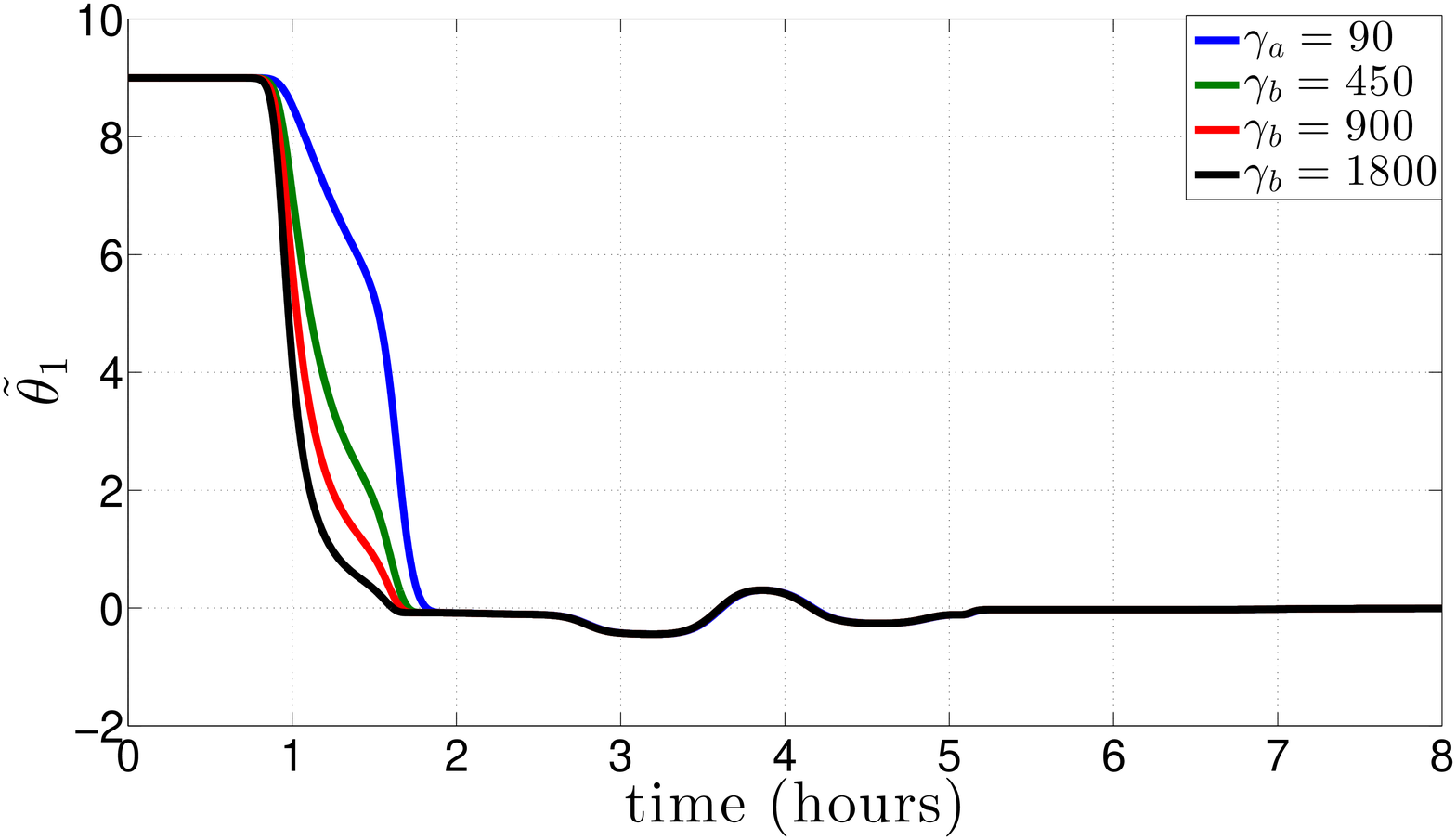}
    }
 {
    \includegraphics[width=0.45\textwidth]{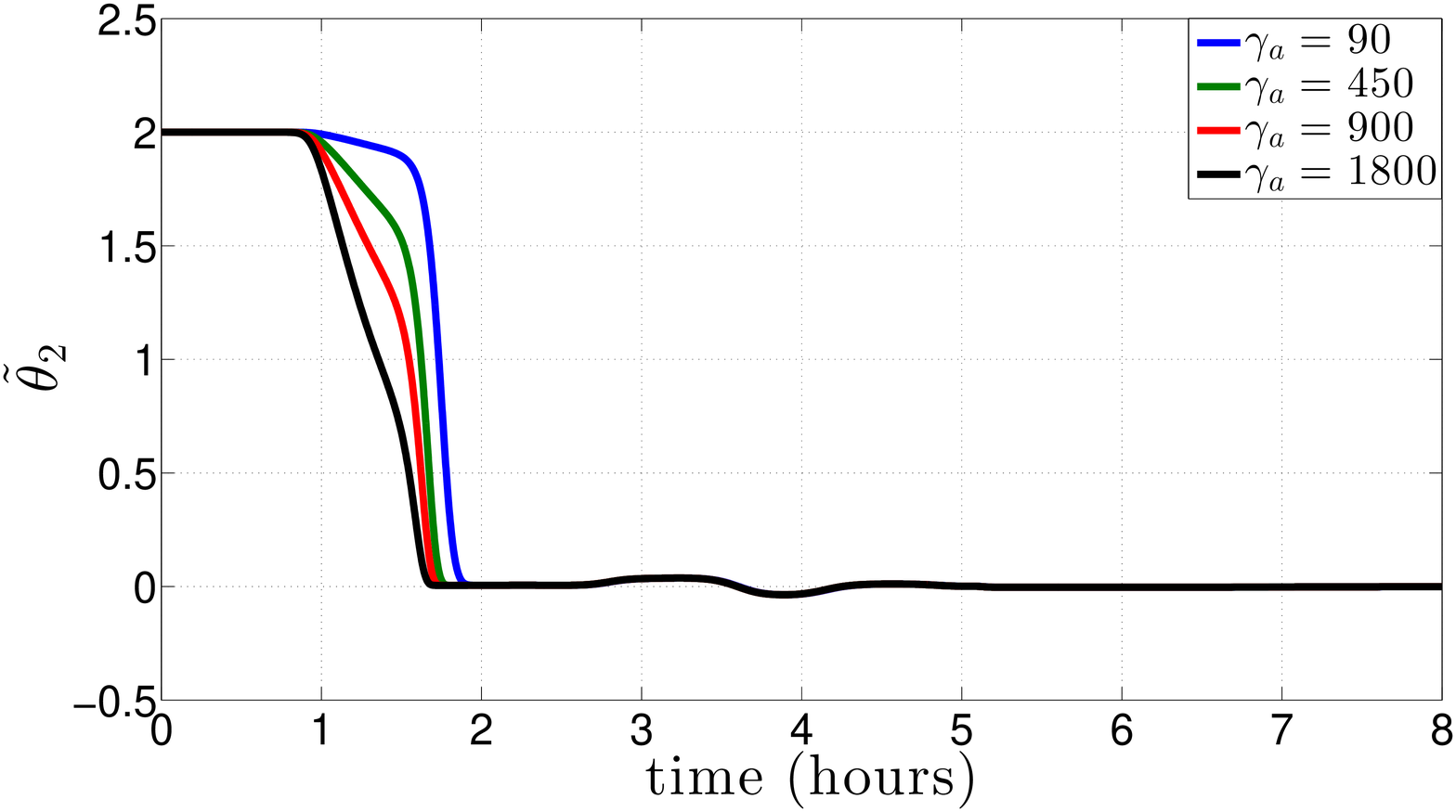}
    }
  {
    \includegraphics[width=0.45\textwidth]{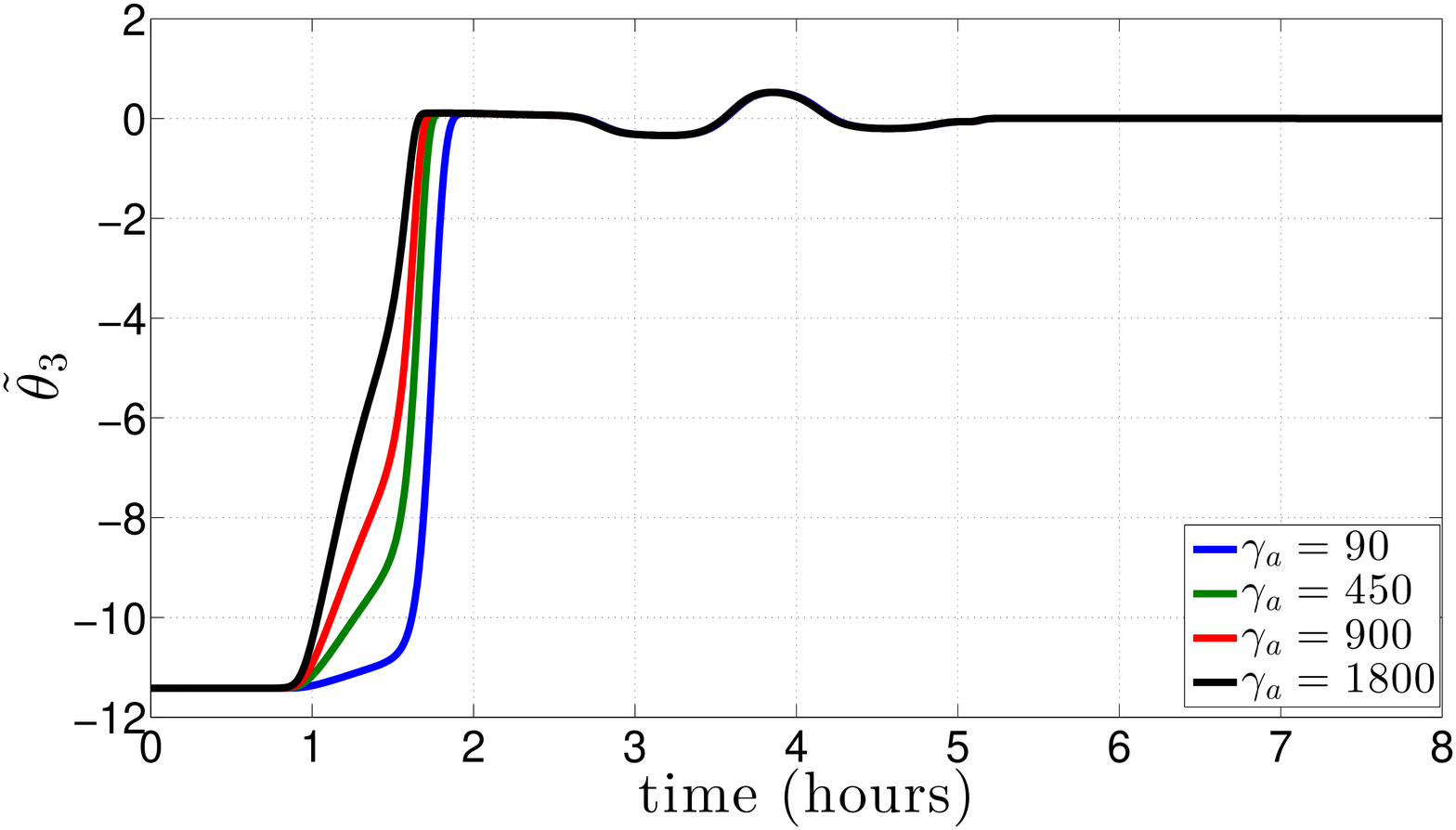}
    }
 {
    \includegraphics[width=0.45\textwidth]{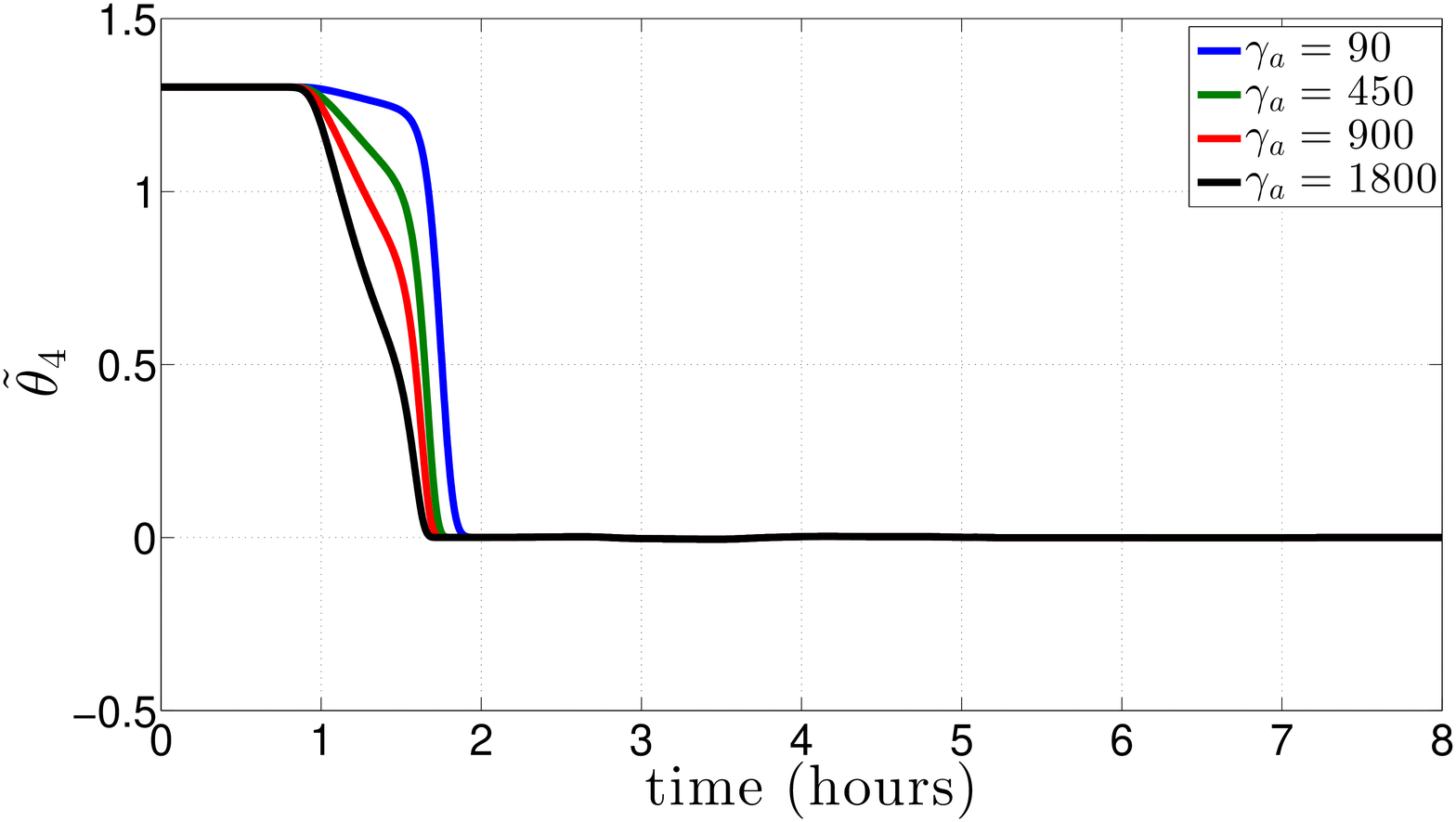}
    }
     {
    \includegraphics[width=0.45\textwidth]{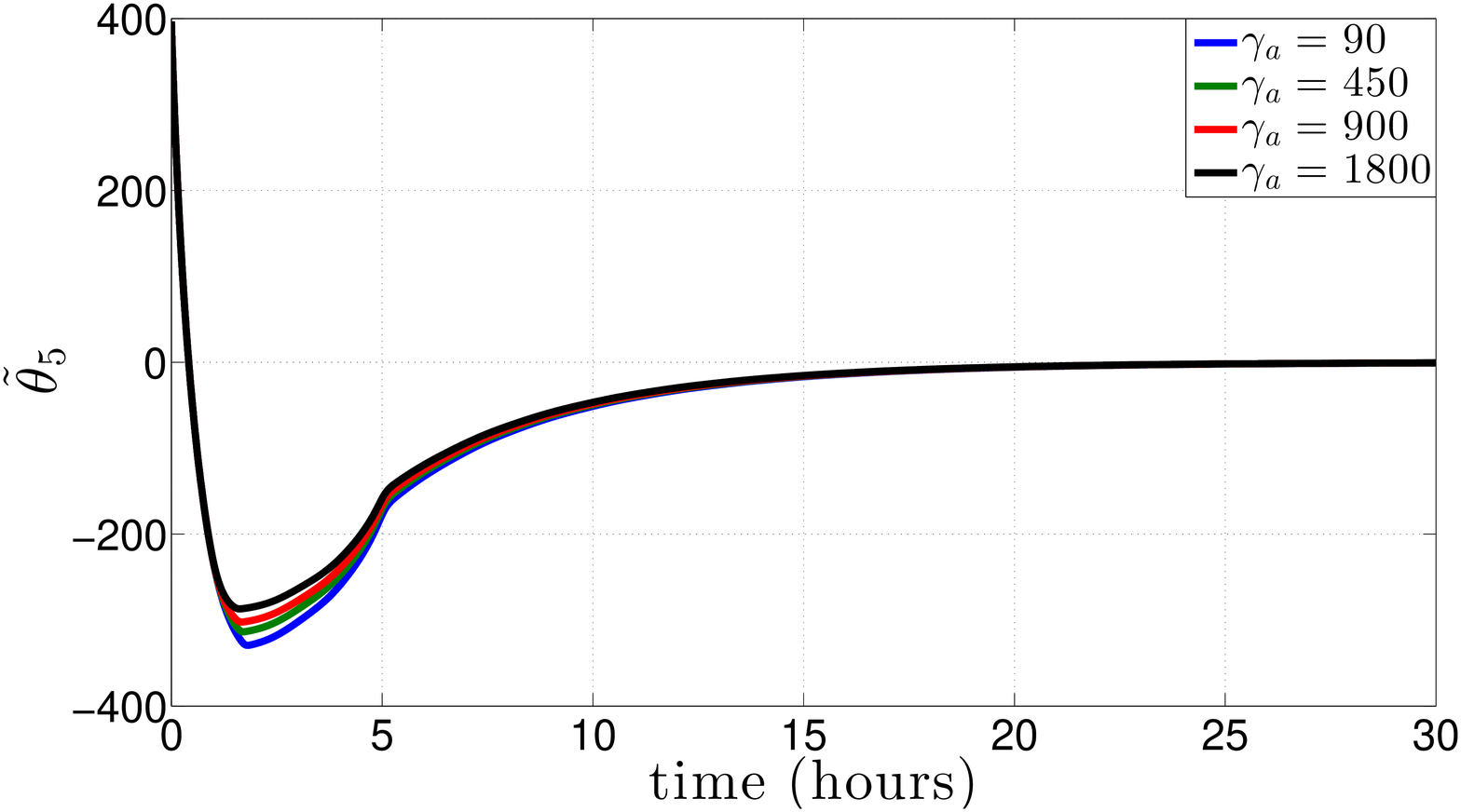}
    }    
 \caption{Transient behavior of the estimation errors $\tilde \theta_{1,5}$ with different values $\gamma_a$ and $T_{w}=350-20\exp(-0.001t)\cos(4t)$.}
    \label{Tw1}
\end{figure}

To assess the effects of the adaptation gain $\gamma_b$ an additional simulation was carried out fixing $\gamma_a=1800$, $T_{in}$ as \eqref{Tin} and $T_w$  as \eqref{tw}. Notice that the behavior of the I\&I estimator does not influence the estimates $\hat \theta_{1,4}$, for this reason, we only show the  transient behavior of the  estimation error $\tilde \theta_5$. Fig. \ref{fig3} shows estimation error $\tilde \theta_5$ for different values of $\gamma_b$  distinguished by the line color in the label of the figure.

\begin{figure}[htp]
\centering
  \includegraphics[width=.85\linewidth]{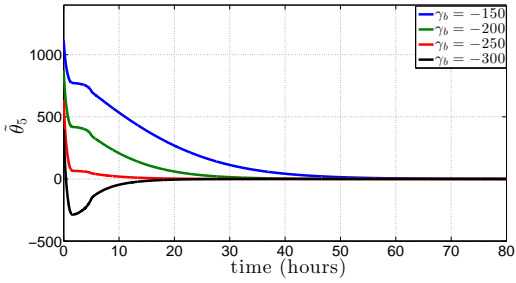}  
  \caption{Transient behavior of the estimation error  $\hat \theta_5$ for different values of $\gamma_b$.}
 \label{fig3}
\end{figure}

Finally, we evaluated the effect in the I\&I estimator of  errors in the {\em a priori} fixed value of the thermal factor $k_0$. The simulation scenario was the same as the one above.   Fig. \ref{fig4} presents $\tilde \theta_5$ assuming that $k_0$ has a measurable error, whose value  is distinguished by the line color in the label of the figure. From the figure we see that the imprecise knowledge of $k_0$ induces a steady-state bias in the estimate, which increases with the size of the error on $k_0$.  

\begin{figure}[htp]
\centering
  \includegraphics[width=.85\linewidth]{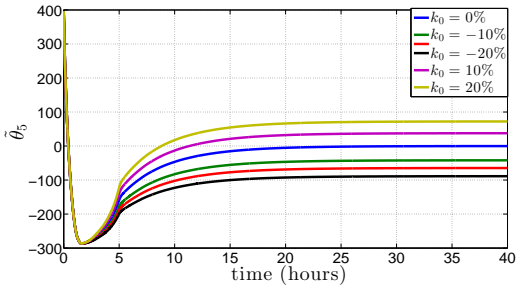}  
  \caption{Transient behavior of the estimation error  $\hat \theta_5$ with incorrect values of $k_0$.}
 \label{fig4}
\end{figure}

\section{Concluding Remarks}
\lab{sec5}
%
We have presented in this paper the first solution to the problem of on-line estimation of the parameters of the classical model of a CSTR given in \eqref{sys}  with the following assumptions: 
\begite 
\item the reactor state---that is, the product concentration and the temperature---are measurable; 
\item  the kinetic constant $k_0$ is known;
\item  the regressor vector \eqref{varphi} satisifes the (extremely weak) {\bf Assumption \ref{ass2}} of IE.
\endite 
It is important to underscore the dramatic difference between IE and the assumption of {\em persistent excitation} \cite[Section 2.5]{SASBODbook}, usually invoked in state observation and identification problems. Also, notice that in \cite{DRGetal}, where a neural network technique is used to approximate the behavior of the CST, the authors require 1000 step changes (!) in the heat exchanger temperature to obtain a reasonable approximation.  Similarly, in \cite{GUADOCPER} a dither signal, consisting of a large sum of sinusoids, is injected to the CSTR to enforce a persistent excitation assumption needed to ensure convergence of an extremum seeking controller. 

Current research is under way to relax the assumption of {\em known} $k_0$ that, as shown in the simulations, induces a non-negligible steady-state error in the estimates. Notice that, with the definition $\theta_6:=\ln k_0$ it is possible to write \eqref{newregequ} in the form
$$
\dot T = \zeta_1 - \zeta_3 e^{\theta^\top_{5,6} \psi},
$$
with the definitions
\begalis{
\zeta_3 & := \theta_3 C_A,\;\psi := \begmat{ -{1 \over T} \\ 1}.
}
Interestingly, mimicking the procedure of Lemma \ref{lem2}, it is possible to design an I\&I estimator for $\theta_{5,6}$ that ensures $\liminf \tilde \theta^\top_{5,6}(t) \psi(t) =0$, independently of the excitation properties of the vector $\psi$. Unfortunately, and rather surprisingly, it can be shown that the equilibrium associated to $\tilde \theta_{5,6}=0$ is always {\em unstable}! Alternative designs of the estimator to overcome this problem are now being explored and we expect to be able to report them in the near future. \\

\begcen {\bf \normalsize Credit authorship contribution statement} \endcen
All authors contributed equally to the paper.


\end{document}